\documentclass{scrartcl}

\usepackage[theoremdefs,final,notex4ht]{latexdev}
\usetikzlibrary{calc}
\usepackage{unixode}

\usepackage[numbers]{natbib} 

\usetikzlibrary{calc}
\usetikzlibrary{matrix}
\tikzstyle{commdiag}=[matrix of math nodes, row sep=3em, column sep=5.5em, text height=1.5ex, text depth=0.25ex,ampersand replacement=\&]
\tikzset{>=stealth}

\usepackage{standalone}
\usepackage{subfig}

\usepackage[affil-sl]{authblk} 

\title{Geometry of discrete-time spin systems}
\author[1]{Robert I.\ McLachlan\thanks{\href{mailto:r.mclachlan@massey.ac.nz}{r.mclachlan@massey.ac.nz}}}
\author[2]{Klas Modin\thanks{\href{mailto:klas.modin@chalmers.se}{klas.modin@chalmers.se}}}
\author[3]{Olivier Verdier\thanks{\href{mailto:olivier.verdier@math.umu.se}{olivier.verdier@math.umu.se}}}

\affil[1]{
	Institute of Fundamental Sciences, Massey University, New Zealand
}
\affil[2]{
	Department of Mathematical Sciences, Chalmers University of Technology and the University of Gothenburg, Sweden
}
\affil[3]{
	Department of Mathematics and Statistics, University of Umeå, Sweden
}

\date{\today}

\IfFileExists{./definitions.tex}{
\usepackage{amsmath,amssymb}

\DeclareMathOperator{\real}{Re}

\newcommand{\ii}{\mathrm{i}}
\newcommand{\pair}[2]{\left\langle #1, #2 \right\rangle}

\providecommand{\abs}[1]{\lvert#1\rvert}
\providecommand{\vect}[1]{\boldsymbol{#1}}

\newcommand{\ud}{\mathrm{d}}

\newcommand{\pd}{\partial}

\newcommand{\CC}{{\mathbb C}}

\newcommand{\Xcal}{\mathfrak{X}}

\newcommand*\SO{\mathrm{SO}}
\newcommand*\so{\mathfrak{so}}

\newcommand*\su{\mathfrak{su}}

\newcommand*\hopf{\pi}

}{

}
\newcommand\RR{\mathbf{R}}

\newcommand\inv{^{-1}}
\newcommand{\w}{\vect{w}}

\newcommand{\midp}[1]{\widetilde{#1}}

\newcommand{\wminus}{\vect{w}}
\newcommand{\wplus}{\vect{W}}
\newcommand{\wmid}{\midp{\wminus}}

\newcommand{\z}{\vect z}
\newcommand{\zminus}{\vect{z}}
\newcommand{\zplus}{\vect{Z}}
\newcommand{\zmid}{\midp{\zminus}}

\NewDocumentCommand\Rp{}{\RR_+}
\NewDocumentCommand\HH{s}{\mathbf{H}\IfBooleanTF{#1}{_*}{}}
\NewDocumentCommand\Hb{s}{{V}\IfBooleanTF{#1}{_*}{}}
\RenewDocumentCommand\CC{s}{\mathbf{C}\IfBooleanTF{#1}{_*}{}}
\NewDocumentCommand\conj{O{z}}{\overline{\vect #1}}
\NewDocumentCommand\jj{}{\mathrm{j}}
\NewDocumentCommand\kk{}{\mathrm{k}}

\NewDocumentCommand\projS{}{\rho}

\begin{document}

\maketitle

\begin{abstract}
Classical Hamiltonian spin systems are continuous dynamical systems on the symplectic phase space $(S^2)^n$.
In this paper we investigate the underlying geometry of a time discretization scheme for classical Hamiltonian spin systems called the \emph{spherical midpoint method}.
As it turns out, this method displays a range of interesting geometrical features, that yield insights and sets out general strategies for geometric time discretizations  of Hamiltonian systems on non-canonical symplectic manifolds.
In particular, our study provides two new, completely geometric proofs that the discrete-time spin systems obtained by the spherical midpoint method preserve symplecticity.

The study follows two paths.
First, we introduce an extended version of the Hopf fibration to show that the spherical midpoint method can be seen as originating from the classical midpoint method on $T^*\RR^{2n}$ for a \emph{collective Hamiltonian}.
Symplecticity is then a direct, geometric consequence.
Second, we propose a new discretization scheme on Riemannian manifolds called the \emph{Riemannian midpoint method}.
We determine its properties with respect to isometries and Riemannian submersions and, as a special case, we show that the spherical midpoint method is of this type for a non-Euclidean metric.
In combination with Kähler geometry, this provides another geometric proof of symplecticity.

 

\end{abstract}

\newpage
\tableofcontents

\section{Introduction} 
\label{sec:introduction}

A well-known integrable PDE is the continuous classical Heisenberg equation of ferromagnetics
\begin{equation}\label{eq:spin_continuous_pde}
	\dot{\w} = \w\times \w'', \quad \w\colon S^1\to S^2,
\end{equation}
where we represent elements in $S^2$ as unit vectors in $\RR^{3}$.
Spatial discretization of this equation by
\begin{equation}\label{eq:approx}
	\w''(s)\approx \frac{\w(s-\Delta s)-2\w(s)+\w(s+\Delta s)}{\Delta s^2}
\end{equation}
leads to the classical Heisenberg spin chain
\begin{equation}\label{eq:spin_chain}
	\dot{\w}_i = \w_i\times(\w_{i-1}+\w_{i+1}), \quad \w_i\in S^2, \quad\w_{0} = \w_n.
\end{equation}
More generally, a \emph{classical Hamiltonian spin system} is of the form
\begin{equation}\label{eq:gen_spin_system}
	\dot{\w_i} = \w_i\times \frac{\pd H}{\pd \w_i}, \quad \w_i\in S^2, \quad i =1,\ldots,n,
\end{equation}
for some Hamiltonian $H\colon (S^2)^n \to \RR$.
In addition to ferromagnetics, examples include the free rigid body and the motion of $n$ point vortices on the sphere.
A key property of the flow of~\eqref{eq:gen_spin_system} is preservation of the symplectic structure of $(S^2)^n$.
In the literature on spin systems it is common to think of the Hamiltonian as a function $H\colon \RR^{3n}\to \RR$.
Notice, however, that only its restriction to $(S^2)^n$ affects the dynamics.

We are interested in symplectic, discrete-time versions of the equations~\eqref{eq:gen_spin_system}.
Towards this goal, a fruitful approach is to regard the two-sphere $S^2$ as a coadjoint orbit of the Lie--Poisson manifold $\so(3)^*$ corresponding to the Lie algebra $\so(3)$ of skew-symmetric matrices (for details on Lie--Poisson manifolds, see \cite{MaRa1999} and references therein).
Then one can use variational discretizations, as those developed by \citet{MoVe1991} for some classical integrable systems, particularly the free rigid body (see also \cite{McZa2005,McMoVeWi2013} for the extension to arbitrary Lie--Poisson manifolds and Hamiltonians).
This \emph{discrete Moser--Veselov (DMV) algorithm} is formulated as an $\SO(3)$-symmetric symplectic map on the phase space $T^*\SO(3)$.
The symmetry implies that the flow descends to a flow on $T^*\SO(3)/\SO(3)\simeq \so(3)^*$, but an explicit representation on $\so(3)^*$, without using auxiliary variables in $\SO(3)$, is not available.
Therefore, an open problem has been to find a minimal-coordinate symplectic discretization of the equations~\eqref{eq:gen_spin_system}.
A solution is provided by the \emph{spherical midpoint method}, first communicated in~\cite{McMoVe2014a}.
This method is given by the map $(\wminus_1,\ldots,\wminus_n)\mapsto(\wplus_1,\ldots,\wplus_n)$ implicitly defined by
\begin{equation}\label{eq:short_method}
	\frac{\wplus_i - \wminus_i}{\Delta t} = \frac{\wminus_i+\wplus_i}{\abs{\wminus_i+\wplus_i}}\times \frac{\pd H}{\pd \w_i}\left( \frac{\wminus_1+\wplus_1}{\abs{\wminus_1+\wplus_1}},\ldots,\frac{\wminus_n+\wplus_n}{\abs{\wminus_n+\wplus_n}} \right).
\end{equation}
A direct proof of its symplecticity is given in~\cite{McMoVe2014b}, where also several examples for specific Hamiltonians are given.
The proof in~\cite{McMoVe2014b} is via a lengthy direct calculation that is not too enlightening.
In this paper we carry out an in-depth geometric investigation of the method~\eqref{eq:short_method} and the corresponding discrete-time spin systems.

Although there has been extensive interest in Lie--Poisson integration, and in the associated discrete mechanics \cite{zhong1988lie,mclachlan1993explicit,reich1994momentum,mclachlan1995equivariant,marsden1999discrete,krysl2005explicit,vankerschaver2007euler}, all previous methods have been closely related to the classical generating functions defined on symplectic vector spaces, and all use extra variables. 
The map defined by~\eqref{eq:short_method} uses no extra variables and is in some sense the first generalization of the Poincar\'e generating function~\cite[vol.~III, \S319]{Po1892} (corresponding to the classical midpoint method) to a noncanonical, nonlinear phase space. 
Earlier discrete Lie--Poisson mechanics has also led to interesting new discrete integrable systems, including the hugely influential Moser--Veselov system~\cite{bloch1998discrete,cardoso2003moser,hairer2006preprocessed}. 
The method~\eqref{eq:short_method} applied to the free rigid body leads to an integrable mapping of an apparently new type. 
These considerations motivate our study of the geometry and discrete mechanics associated with the spherical midpoint method.

First, in \autoref{sec:proof_2}, we study the symplectic geometry.
In particular, we show that the spherical midpoint method can be interpreted as a \emph{collective symplectic integrator}, such as developed in~\cite{McMoVe2014c,McMoVe2014d}.
We establish this connection by an extension of the classical Hopf fibration.
In addition to geometric insights, the connection to collective integrators also establishes an independent geometric proof of symplecticity, completely different from the direct proof in~\cite{McMoVe2014b}.

Second, in \autoref{sec:riemann_midpoint}, we study the Riemannian geometry.
The classical midpoint method evaluates the vector field at the midpoint of a straight line joining the start and ending points. 
This suggests a generalization to Riemannian manifolds, and a \emph{Riemannian midpoint method}, that appears to be new. 
We introduce this method and establish some of its basic properties, including equivariance with respect to the isometry group of the manifold and natural behaviour with respect to Riemannian submersions.
Perhaps counterintuitively, the Riemannian midpoint method for the standard Riemannian structure on $S^2$ is not symplectic.
Nevertheless, there is a Riemannian metric for which the corresponding Riemannian midpoint method is exactly the spherical midpoint method.
We arrive at this result by examining how the Kähler structure of the space of quaternions relates to the extended Hopf map.
This provides another way to view the spherical midpoint method, and yet another proof of symplecticity, based on Kähler geometry.

We use the following notation.
$\Xcal(M)$ denotes the space of smooth vector fields on a manifold $M$.
If $M$ is a Poisson manifold, and $H\in C^{\infty}(M)$ is a smooth function on $M$, then the corresponding Hamiltonian vector field is denoted $X_H$.
Let us also recollect the concept of \emph{intertwining}.
To this extent, let $M$ and $N$ be two manifolds, and consider a differentiable map $f\colon N\to M$.
We say that $f$ \emph{intertwines} $X\in\Xcal(M)$ and $Y\in\Xcal(N)$ if $X\circ f = Tf\circ Y$.
(Some authors prefer to say that $X$ and $Y$ are \emph{$f$-related}.)
Likewise, we say that $f$ intertwines a function $\Phi \colon M \to M$ and a function $\Psi \colon N \to N$ if
\begin{equation}\label{eq:rho_related}
	\Phi\circ f = f\circ \Psi
	.
\end{equation}
Finally, the Euclidean length of a vector $\w\in\RR^{d}$ is denoted $\abs{\w}$.
If $\w\in\RR^{3n}\simeq (\RR^{3})^n$, then $\w_i$ denotes the $i$:th component in $\RR^{3}$.

We continue this section with a concise presentation of the spherical midpoint method.

\subsection{Spherical midpoint method}\label{sub:smm}

Here we review some background on the spherical midpoint method~\eqref{eq:short_method}.
All the results in this section are also available in~\cite{McMoVe2014b}.

A key point in this paper is the relation between the spherical midpoint method and the classical midpoint method on vector spaces.
We recall its definition.
\begin{definition}\label{def:classical_midpoint}
	Let $X$ be a vector field defined on an open subset of a vector space.
	The \textbf{classical midpoint method} for $X$ is the mapping $\zminus\mapsto\zplus$ defined by
	\begin{equation}\label{eq:implicit_midpoint}
		\frac{\zplus-\zminus}{\Delta t} = X\Big(\frac{\zplus+\zminus}{2}\Big),
	\end{equation}
	where $\Delta t>0$ is the time step.
\end{definition}

The vector field $X_H$ given by the right-hand side of~\eqref{eq:gen_spin_system} is defined on $(S^2)^n$ (since the Hamiltonian $H$ is defined on $(S^2)^n$).
In order to relate the spherical to the classical midpoint method we need to extend the vector field $X_H$ to $(\RR^{3}\backslash\{0\})^n$.
For this, we define a projection map $\projS$ by
\begin{equation}\label{eq:projection_map_various_r}
	\projS(\w) = \Big(\frac{\w_1}{\abs{\w_1}},\ldots,\frac{\w_n}{\abs{\w_n}}\Big).
\end{equation}
The basic observation is then that the spherical midpoint method~\eqref{eq:short_method} can be written
\begin{equation}\label{eq:classical_spherical_relation}
	\frac{\wplus-\wminus}{\Delta t} = \big( X_H\circ\projS\big)\left( \frac{\wplus+\wminus}{2} \right).
\end{equation}
Comparing with~\eqref{eq:implicit_midpoint}, we see that \emph{the spherical midpoint method is the classical midpoint method applied to $X_H\circ\projS$}.
This observation is the starting point for our developments.

What is then characteristic for vector fields of the form $\xi\circ\projS$ with $\xi\in\Xcal\big((S^2)^n\big)$?
This question leads us to the next cornerstone in the paper.

\begin{definition}\label{def:rays}
	The \textbf{ray} through a point $\w\in (\RR^{3}\backslash\{0\})^n$ is the subset
	\begin{equation}
		\{ (\lambda_1\w_1,\ldots,\lambda_n\w_n); \vect\lambda\in \RR_+^{n} \}.
	\end{equation}
\end{definition}

The set of all rays is in one-to-one relation with $(S^2)^n$.
Note that the vector field~$X=\xi\circ\projS$ is \emph{constant on rays}.
The property of being constant on rays is passed on from Hamiltonian functions to Hamiltonian vector fields.

\begin{lemma}[\cite{McMoVe2014b}]
	\label{lma:XHconstant}
	If a Hamiltonian function $H$ on $(\RR^{3}\backslash\{0\})^n$ is constant on rays, then so is its Hamiltonian vector field $X_H$, defined by
	\begin{equation}\label{eq:Ham_vf}
		X_{H}(\w) = \sum_{k=1}^{n} \w_k\times \frac{\pd H(\w)}{\pd \w_k} .
	\end{equation}
\end{lemma}

The implication is that we may replace $S^2$ with the manifold of rays, and Hamiltonian functions on $(S^2)^n$ with Hamiltonian functions on $\RR^{3}\backslash\{0\}$ that are constant on rays.
In this representation, the spherical midpoint method becomes the classical midpoint method, as we have seen.

$\RR^{3n}$ and $(\RR^{3}\backslash\{0\})^n$ are Poisson manifolds with the Poisson bracket
\begin{equation}\label{eq:LiePoisson_bracket}
	\{F,G \}(\w) = \sum_{k=1}^{n}\big(
		\frac{\pd F(\w)}{\pd \w_k}\times 
		\frac{\pd G(\w)}{\pd \w_k}
	\big)\cdot \w_k.
\end{equation}
This is the canonical Lie--Poisson structure of $(\so(3)^{*})^{n}$, or $(\su(2)^{*})^{n}$, obtained by identifying $\so(3)^{*}\simeq \RR^{3}$, or $\su(2)^{*}\simeq \RR^{3}$.
For details, see \cite[\S\!~10.7]{MaRa1999} or \cite{McMoVe2014c}.

The flow of a Hamiltonian vector field $X_H$ on $\RR^{3n}$, denoted $\exp(X_{H})$, preserves the Lie--Poisson structure, i.e.,
\begin{equation}\label{eq:preservation_of_LP}
	\{F\circ\exp(X_{H}),G\circ\exp(X_{H}) \} = \{F,G \}\circ\exp(X_{H}), \quad \forall F,G\in C^{\infty}(\RR^{3n}).
\end{equation}
The flow $\exp(X_{H})$ also preserves the \emph{coadjoint orbits}~\cite[\S\!~14]{MaRa1999}, given by
\begin{equation}\label{eq:coadjoint_orbits}
	S_{\lambda_1}^{2}\times\cdots\times S_{\lambda_n}^2\subset \RR^{3n},\quad \vect \lambda = (\lambda_1,\ldots,\lambda_n)\in (\RR^{+})^{n},
\end{equation}
where $S^{2}_{\lambda}$ denotes the 2--sphere in $\RR^3$ of radius~$\lambda$.
A \emph{Lie--Poisson integrator} for $X_{H}$ is a method that, like the exact flow, preserves the Lie--Poisson structure and the coadjoint orbits.

It is possible to extend the spherical midpoint method so that it encompasses all non-singular coadjoint orbits, instead of only the one with radius one.
Define the map $\Gamma(\wminus,\wplus)$ by
\begin{equation}\label{eq:Gamma_projection}
	\Gamma\colon \big(\wminus,\wplus\big) \longmapsto
	\Big( \frac{\sqrt{\abs{\wminus_{1}}\abs{\wplus_{1}}}(\wminus_{1}+\wplus_{1})}{\abs{\wminus_{1}+\wplus_{1}} },\ldots, \frac{\sqrt{\abs{\wminus_{n}}\abs{\wplus_{n}}}(\wminus_{n}+\wplus_{n})}{\abs{\wminus_{n}+\wplus_{n}} } \Big).
\end{equation}
We then have the following definition.

\begin{definition}\label{def:extended_spherical_midpoint}
	Let $X$ be a vector field on an open subset of $\RR^{3n}$.
	The \textbf{extended spherical midpoint method} for $X$ is the discrete-time system $\wminus\mapsto\wplus$ defined by
	\begin{equation}\label{eq:area_midpoint_LP}
		\frac{\wplus-\wminus}{\Delta t} = X\big(\Gamma(\wminus,\wplus) \big).
	\end{equation}
\end{definition}

A consequence of the geometric investigations in \autoref{sec:proof_2} is that the method~\eqref{eq:area_midpoint_LP} is a Lie--Poisson integrator, directly related to the classical midpoint method on $T^*\RR^{2n}$ through the concept of \emph{collective symplectic integrators}.
A consequence of the geometric investigations in \autoref{sec:riemann_midpoint} is that the method~\eqref{eq:area_midpoint_LP} is a \emph{Riemannian midpoint method} with respect to a non-standard metric on $(\RR^{3}\backslash\{0\})^n$.

\section{Symplectic and Poisson geometry} 
\label{sec:proof_2}

In this section we show that the extended spherical midpoint method and the classical midpoint method are coupled through \emph{collective symplectic integrators}~\cite{McMoVe2014c,McMoVe2014d}.

We use quaternions as it makes the calculations more transparent; the field of quaternions is denoted $\HH$.
We apply the convention that product sets $\CC^n$ and $\HH^n$ inherit the componentwise operations of the underlying field.
For instance, if $\vect z = (z_1,z_2) \in \CC^2$, then 
\begin{align}
	\label{eq:componentwiseexample}
	\vect z^3 = (z_1^3, z_2^3)
	.
\end{align}
All operations are defined in the same manner.

\newcommand*\dc{\varpi}
\subsection{Intertwining by the double covering map}

We first consider intertwining in the double covering case.
We define the \emph{double covering map}
\begin{align}
	\label{eq:doublecovering}
	\dc \colon \CC^{n} \to \CC^{n}, \quad \vect z \mapsto \vect z^2
\end{align}
following the convention in \eqref{eq:componentwiseexample}.
We use the notation $\CC* \coloneqq \CC\backslash \{0\}$ and $\CC*^n \coloneqq (\CC*)^n$.

\begin{lemma}
	\label{lma:coveringmidpoint}
	Let $X,Y\in \Xcal(\CC*^n)$ and let $\Phi(\Delta t X)$ and $\Phi(\Delta tY)$ denote the classical midpoint method~\eqref{eq:implicit_midpoint} on~$\CC^{n}$ for $X$ and $Y$ respectively.
	Assume that:
	\begin{enumerate}
		\item $X(\vect\lambda \vect z) = X(\vect z)$ for all $\vect\lambda \in \Rp^n$, i.e., $X$ is constant on rays.
		\item $Y$ is tangent to the tori in $\CC_*^n$, i.e., $Y(\vect{z})/\vect{z}$ is imaginary for all $\vect{z}\in \CC*^n$.
		\item $\dc$ intertwines $X$ and $Y$
	\end{enumerate}
	Then $\dc$ intertwines $\Phi(\Delta t X)$ and $\Phi(\Delta t Y)$.
\end{lemma}

\begin{proof}
	The proof is illustrated in \autoref{fig:coveringmidpoint}.
	Consider two points $\zminus$ and $\zplus$, solutions of one step of the classical midpoint method for $Y$.
	The midpoint is $\zmid \coloneqq (\zminus+\zplus)/2$, so
	\begin{align}
		\zplus - \zminus = \Delta t Y(\zmid).
	\end{align}
	The assumption  that $\dc$ intertwines $X$ and $Y$ is
	\begin{align}
		X(\vect z^2) = 2 \vect z Y(\vect z).
	\end{align}
	We have
	\begin{equation}\label{eq:midpoint_c2}
	\begin{split}
		\zplus^2 - \zminus^2 &= 2 \zmid (\zplus-\zminus) \\
		&= 2 \zmid Y(\zmid) \\
		&= \Delta t X(\zmid^2).	
	\end{split}
	\end{equation}
	Consider the general identity
	\begin{align}
		\zminus^2 + \zplus^2 = \frac{1}{2} \paren[\big]{(\zplus-\zminus)^2 + (\zminus + \zplus)^2}.
	\end{align}
	Without loss of generality, we assume that $\zmid\in\RR^n$.
	We assumed that $Y$ was tangent to circles, so $Y(\zmid)$ has only imaginary components, so the same holds for $\zplus - \zminus$, which implies that $(\zplus-\zminus)^2$ only has real components.
	We therefore obtain that $\zminus^2 + \zplus^2$ is in $\RR^n$.
	Since $X$ is constant on the rays, and since $\RR^{n}$ is a ray in $\CC_*^{n}$, we get
	\begin{align}
		X\big((\zminus^2+\zplus^2)/2\big) = X(\zmid^2).
	\end{align}
	From~\eqref{eq:midpoint_c2} we now have
	\begin{equation}\label{eq:midpoint_fulfilled_c2}
		\zplus^2 - \zminus^2 = \Delta t X\big((\zminus^2+\zplus^2)/2\big).
	\end{equation}
	This proves the result.
\end{proof}

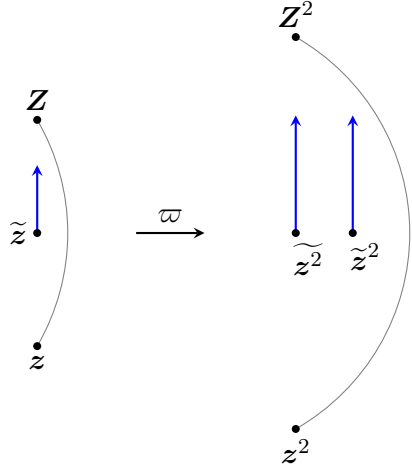
\begin{figure}
	\centering	
	\begin{tikzpicture}[
		scale=3,
		vector/.style={draw, ->, thick, blue},
		arccircle/.style={draw, thin, gray},
		]
		\newcommand\myangle{30} 
		\newcommand\lenf{.3} 
		\newcommand\dmyangle{60} 
		\newcommand\zmsq{.75} 
		\newcommand\leng{.52} 
		\newcommand\drawpoint[1]{\fill[black] (#1) circle (.5pt);}


		\path[arccircle] (1,0) arc (0:\myangle:1) node (z1){};
		\path[arccircle] (1,0) arc (0:-\myangle:1) node (z0) {};
		\node[below] at (z0) {$\zminus$};
		\node[above] at (z1) {$\zplus$};
		\coordinate (zm) at ($(z0)!.5!(z1)$);
		\node[left] at (zm) {$\zmid$};
		\path[vector] (zm) -- +(0,\lenf);
		\drawpoint{zm}
		\drawpoint{z0}
		\drawpoint{z1}
		\begin{scope}[xshift=1.5cm]
		\path[arccircle] (1,0) arc (0:\dmyangle:1) node (sz1){};
		\path[arccircle] (1,0) arc (0:-\dmyangle:1) node (sz0) {};
		\node[below] at (sz0) {$\zminus^2$};
		\node[above] at (sz1) {$\zplus^2$};
		\coordinate (szm) at (.75,0);
		\node[below, xshift=1ex] at (szm) {$\zmid^2$};
		\path[vector] (szm) -- +(0.,\leng);
		\coordinate (szM) at ($(sz0)!.5!(sz1)$);
		\node[below left, xshift=3ex] at (szM) {$\midp{\z^2}$};
		\path[vector] (szM)  -- +(0.,\leng);
		\drawpoint{sz0}
		\drawpoint{sz1}
		\drawpoint{szm}
		\drawpoint{szM}
		\end{scope}
		\path[draw,->,thick, black] (1.3,0) -- (1.6,0) node[above, midway]{$\dc$};
	\end{tikzpicture}
	\caption{An illustration of \autoref{lma:coveringmidpoint}.
	The midpoint $\zmid$ maps to the same ray as the midpoint $\midp{\z^2}\coloneqq (\zminus^2+\zplus^2)/2$.
	}
	\label{fig:coveringmidpoint}
\end{figure}

\subsection{Intertwining by the extended Hopf map}

Consider the map
\begin{align}
	\hopf \colon \HH^n &\to \HH^n\\
	\vect z &\mapsto \frac{1}{4}\vect z \kk \conj
\end{align}
Again, we follow the convention of \eqref{eq:componentwiseexample} and all the operations are defined componentwise.
Note that the image of $\hopf$ has no real part.
Let us define the three dimensional subspace of \emph{pure imaginary quaternions} by
\begin{align}\label{eq:imaginary_quaternions}
	\Hb \coloneqq  \operatorname{span}\set{\ii, \jj, \kk}
\end{align}
If we identify $\Hb$ with $\RR^3$, we can regard $\hopf$ as a map
\begin{align}\label{eq:extended_Hopf_quaternions}
	\hopf \colon \HH^n \to \RR^{3n}.
\end{align}
When $n=1$ this is the \emph{extended Hopf map}, essential in the construction of collective Lie--Poisson integrators on $\RR^{3}$~\cite{McMoVe2014c}.

$\Hb^n$ is naturally endowed with the Lie--Poisson structure~\eqref{eq:LiePoisson_bracket}, inherited from~$\RR^{3n}$.
We use the notation $\Hb* \coloneqq \Hb\backslash\{0 \}$ and $\Hb*^n \coloneqq (\Hb*)^n$.

\begin{lemma}
	\label{lma:Hopfmidpoint}
	Let $X\in\Xcal(\Hb*^n)$ and $Y\in\Xcal(\HH*^n)$, and let $\Phi(\Delta t X)$ and $\Psi(\Delta t Y)$ denote the classical midpoint methods on $\Hb^n$ and $\HH^n$ respectively.
	Assume that:
	\begin{enumerate}
		\item $X(\vect\lambda \vect w) = X(\vect w)$ for all $\vect\lambda \in \Rp^n$ and $\vect w \in \Hb*^{n}$.
		\item $Y$ is tangent to 3-spheres, i.e., $\vect z \inv Y(\vect z) \in \Hb^n$ for all $\vect z\in\HH*^{n}$.
		\item $Y$ is orthogonal to the fibres of $\hopf$, i.e. $\kk \vect z\inv Y(\vect z) \in \Hb^n$ for all $\vect z \in \HH*^{n}$.
		\item $\hopf$ intertwines $X$ and $Y$.
	\end{enumerate}
	Then $\hopf$ intertwines $\Phi(\Delta t X)$ and $\Psi(\Delta t Y)$.
\end{lemma}

\begin{proof}
	Consider $\zminus$ and $\zplus$, solution of the classical midpoint method in $\HH^n$ for $Y$.
	The midpoint is denoted by
	\begin{align}
		\zmid \coloneqq \frac{\zminus+\zplus}{2}.
	\end{align}
	Since the classical midpoint method is equivariant with respect to affine transformations, we may, without loss of generality, assume that $\zmid$ is real (i.e., all its components are real).
	At the point $\zmid$, the real direction is orthogonal to the 3--spheres, and $\kk$ is the fibre direction.
	Without loss of generality, we may further assume that $Y(\zmid)$ is proportional to $\ii$, i.e., $Y(\zmid) = \ii \vect a$ and $\vect a\in\RR^n$.
	As a result, the components of $\zminus$ and $\zplus$ belong to $\operatorname{span}\{1,\ii \}$, which we identify with the complex plane $\CC$, so we write $\zminus,\zplus \in \CC^n$.

	Notice that since $\kk \ii = -\ii \kk$,  we have for $z = a + \ii b$:
	\begin{align}
	z \kk = \bar{z} \kk
	\end{align}
	When restricted to $\CC^n$, the Hopf map at $\vect z \in \CC^n$ is thus
	\begin{align}
		\hopf(\z) = \z \kk \bar{\z} = \z^2 \kk
	\end{align}
	This means that $\pi(\CC^n) \subset (\operatorname{span}\{\kk,\jj\})^n$.
	We identify $(\operatorname{span}\{\kk,\jj\})^n$ with $\CC^n$ by $\kk \leftrightarrow 1$ and $\jj \leftrightarrow -\ii$.
	With these identifications, the restriction of $\pi$ to $\CC^{n}$ is the double covering map $\dc$ defined in \eqref{eq:doublecovering}, so the result follows from~\autoref{lma:coveringmidpoint}.
\end{proof}

$\HH*^n$ carries the structure of a symplectic manifold; the Hamiltonian vector field corresponding to $F\in C^\infty(\HH*^n)$ is
\begin{equation}\label{eq:Hamiltonian_vf_quaternions}
	X_F(\z) = \nabla F(\z) k.
\end{equation}
The symplectic structure coincides with the canonical symplectic structure of $T^*\RR^{2n}$ under the identification 
\begin{equation}
	\HH*^n \ni\vect a + \ii\vect b + \jj\vect c + \kk\vect d \longmapsto \big(\underbrace{(\vect b,\vect d)}_{\vect q},\underbrace{(\vect a,\vect c)}_{\vect p}\big)\in T^*\RR^{2n}.	
\end{equation}
Likewise, $\Hb*^n$ carries the structure of a Poisson manifold; the Hamiltonian vector field corresponding to $H\in C^\infty(\Hb*^n)$ is
\begin{equation}\label{eq:Hamiltonian_vf_poisson_quaternions}
	X_H(\vect w) = \Pi\big( \vect w\nabla H(\vect w) \big),
\end{equation}
where $\Pi\colon \HH^n \to \Hb^n$ is the projection $\vect a + \ii\vect b + \jj\vect c + \kk\vect d\mapsto \ii\vect b + \jj\vect c + \kk\vect d$.
Under the identification of $\Hb$ with $\RR^3$, this Poisson structure coincides with the standard Lie--Poisson structure of $\RR^3$, so equation~\eqref{eq:Hamiltonian_vf_poisson_quaternions} is just another way of writing equation~\eqref{eq:Ham_vf}.

We now investigate what the conditions on $X$ and $Y$ in~\autoref{lma:Hopfmidpoint} mean for Hamiltonian vector fields $X=X_H$ and $Y=X_F$.
It follows directly from \autoref{lma:XHconstant} that $X_H$ fulfils condition~1 in \autoref{lma:Hopfmidpoint} if and only if $H$ fulfils the same condition, i.e., $H(\vect\lambda\vect w) = H(\vect w)$ for all $\vect\lambda \in \RR_+^n$ and $\vect w\in \Hb*^n$.
The next result shows that $X_F$ fulfils condition~3 in \autoref{lma:Hopfmidpoint} if and only if $F$ fulfils the same condition.

\begin{lemma}
	\label{lma:constantorthogonal}
	Let $F\in C^{\infty}(\HH*^n)$. 
	Then $X_F$ is orthogonal to the fibres of $\hopf$, i.e., 
	\begin{equation}
		\kk \z\inv X_F(\z) \in \Hb^n		
	\end{equation}
	if and only if 
	\begin{equation}\label{eq:Hconst_on_rays}
		F(\vect\lambda \z) = F(\z),\quad \forall\, \lambda \in \Rp^{n} \text{ and } \z\in \HH*^n.
	\end{equation}
\end{lemma}

\begin{proof}
	From~\eqref{eq:Hamiltonian_vf_quaternions} it follows that $X_F(\z)$ is orthogonal to the fibres if and only if
	\begin{align}
		\kk \z\inv\nabla F(\z) \kk \in \Hb^n.
	\end{align}
	This is equivalent to
	\begin{align}
		\z\inv\nabla F(\z) \in \Hb^n
		,
	\end{align}
	which means that $\z\inv\nabla F(\z)$ is pure imaginary, so $\nabla F(\z)$ is tangential to the spheres.
	Since this is true for any $\z\in \HH*^n$, it means that $F(\vect\lambda \z) = F(\z)$ for $\vect \lambda \in \Rp^n$.
\end{proof}

Given $H\in C^\infty(\Hb*^n)$ we can now construct discretizations of systems on $\Hb*^n$ in two ways:
\begin{enumerate}
	\item The classical midpoint method on $\HH*^n$ for the vector field $X_{H\circ\hopf}$ descends to a Lie--Poisson integrator on $\Hb*^n$.
	The resulting discrete-time system, examined in~\cite{McMoVe2014d}, is an example of a \emph{collective symplectic integrator}.

	\item The extended spherical midpoint method~\eqref{eq:area_midpoint_LP} for the vector field $X_H$ gives a discrete-time system on $\Hb*^n$.
\end{enumerate}
In general the two methods are different.
They do, however, coincide for ray-constant Hamiltonian vector fields, which is the main result of this section.

\begin{theorem}\label{thm:midpoint_relatedness}
	Let $H\in C^{\infty}(\Hb*^n)$ be constant on rays, let $\Phi$ denote the classical midpoint method on $\Hb*^n$, and let $\Psi$ denote the classical midpoint methods on~$\HH*^n$.
	Then the extended Hopf map~$\hopf$ intertwines $\Phi(X_H)$ and $\Psi(X_{H\circ \hopf})$.
	That is, the map on $\Hb*^n$ induced by $\Psi(X_{H\circ \hopf})$ coincides with $\Phi(X_H)$.
\end{theorem}

\begin{proof}
	Appealing to \autoref{lma:XHconstant}, the vector field $X_H$ is constant on rays.
	If $H(\vect\lambda \vect w) = H(\vect w)$ for all $\vect\lambda \in\Rp^n$, then $H\circ\hopf$ fulfils the same property, since $\hopf$ is homogeneous, that is:
	\begin{align}
		\hopf(\vect\lambda \z) = \vect\lambda^2 \hopf(\z), \quad \vect\lambda \in\Rp^n
		.
	\end{align}
	We can therefore use \autoref{lma:constantorthogonal} to obtain that $X_{H\circ\hopf}$ is orthogonal to the fibres. 

	The Hopf map is a Poisson map \cite{McMoVe2014c}, so
	\begin{align}
		T_{\z}\hopf\cdot X_{H\circ\hopf}(\z) =  X_{H}(\hopf(\z)), \quad \z \in \HH,
	\end{align}
	which means that $\hopf$ intertwines $X_H$ and $X_{H\circ\hopf}$.
	The result now follows from \autoref{lma:Hopfmidpoint}, since $\Phi(X_H)$ coincides with the classical midpoint method applied to $X_H$.
\end{proof}

\begin{corollary}\label{cor:symplect}
	The spherical midpoint method, defined by~\eqref{eq:short_method}, is symplectic map and the extended spherical midpoint method, defined by~\eqref{eq:area_midpoint_LP}, is a Lie--Poisson integrator.
\end{corollary}

\begin{proof}
	Follows from \autoref{thm:midpoint_relatedness} and the result in~\cite{McMoVe2014d} that collective symplectic integrators are symplectic.
\end{proof}

\section{Riemannian and Kähler geometry}\label{sec:riemann_midpoint}

In this section we describe the geometry of the extended spherical midpoint method from the viewpoint of Riemannian and Kähler geometry.
More precisely, we construct a method on~$\Hb*^n$, stemming from a non-Euclidean metric, that coincides with the spherical midpoint method for Hamiltonian functions that are constant on rays.
The relation between these two methods is established through the classical midpoint method on~$\HH*^n$.
Before working this out in \autoref{sub:special_case}, we develop, in~\autoref{sub:rmm}, a general theory of midpoint methods on Riemannian and Kähler manifolds.
This theory further reveals the geometry of the spherical midpoint method, and provides a starting point for generalizations.

\subsection{Riemannian midpoint methods}\label{sub:rmm}

Given a Riemannian manifold $(M,\mathsf{g})$, let $[0,1]\ni t\mapsto \gamma_{\mathsf{g}}(t;\wminus,\wplus)\in M$ denote the geodesic curve between $\wminus$ and $\wplus$.

\begin{definition}\label{def:riemannian_midpoint}
Given a vector field $X \in \Xcal(M)$, the \emph{Riemannian midpoint method} on $M$ is the discrete-time system $\Phi_{\mathsf{g}}(\Delta tX)\colon \wminus\mapsto\wplus$ defined by
\begin{equation}\label{eq:Riemannian_midpoint_method}
	\frac{\ud}{\ud t}\Big|_{t=1/2} \gamma_{\mathsf{g}}(t;\wminus,\wplus) = \Delta t X\big(\gamma_{\mathsf{g}}(1/2,\wminus,\wplus)\big).
\end{equation}
\end{definition}
If $M=\RR^{d}$ and $\mathsf{g}$ is the Euclidean metric, then \eqref{eq:Riemannian_midpoint_method} coincides with the definition of the classical midpoint method~\eqref{eq:implicit_midpoint}.

Riemannian midpoint methods transform naturally under change of coordinates:
\begin{proposition}\label{pro:Riemannian_coordinate_free}
	Let $M$ and $N$ be two diffeomorphic manifolds, let $\psi\colon N\to M$ be a diffeomorphism, and let $\mathsf{g}$ be a Riemannian metric on $M$.
	Then
	\begin{equation}\label{eq:coordinte_free_midpoint}
		\psi^{-1}\circ \Phi_{\mathsf{g}}(\Delta tX)\circ\psi = \Phi_{\psi^{*}\mathsf{g}}(\Delta t\psi^{*}X).
	\end{equation}
\end{proposition}

\begin{proof}
	The result follows from the definition~\eqref{eq:Riemannian_midpoint_method} of $\Phi_{\mathsf{g}}$ and standard change of coordinate formulas in differential geometry.
\end{proof}

A consequence of \autoref{pro:Riemannian_coordinate_free} is that Riemannian midpoint methods are equivariant with respect to isometric group actions:

\begin{proposition}\label{pro:metric_equivariance}
	Let $(M,\mathsf{g})$ be a Riemannian manifold, and let $G$ be a Lie group acting isometrically on $M$.
	Then the Riemannian midpoint method $\Phi_{\mathsf{g}}$ is \emph{equivariant} with respect to $G$, i.e.,
	\begin{equation}\label{eq:equivariance_of_Riemann_midpoint}
		\psi_{g^{-1}}\circ \Phi_{\mathsf{g}}(\Delta tX)\circ\psi_{g} = \Phi_{\mathsf{g}}(\Delta t\psi_g^{*}X)
	\end{equation}
	where $\psi_{g}\colon M\to M$ denotes the action map of $g\in G$.
\end{proposition}

\begin{proof}
	The result follows from \autoref{pro:Riemannian_coordinate_free} and $\psi_{g}^*\mathsf{g} = \mathsf{g}$ (the action is isometric).
\end{proof}

We will now discuss a generalised version of \autoref{pro:Riemannian_coordinate_free}, where $M$ and $N$ are no longer diffeomorphic.

Let $\pi\colon N\to M$ be a \emph{submersion} from $N$ to another manifold $M$ ($\pi$ is smooth and its Jacobian matrix $T_{\z}\pi\colon T_{\z} N \to T_{\pi(\z)}M$ is surjective at every $\z\in N$.)
Then~$\pi$ induces a \emph{vertical distribution} $\mathrm{Vert}$ by $\mathrm{Vert}_{\z} = \{\vect v\in T_{\z}N; T_{\z}\pi\cdot \vect v = 0\}$.
By construction, the vertical distribution is integrable, and the fibre through $\z\in N$ is given by $\pi^{-1}(\{\pi(\z) \})$.
If $(N,\mathsf{h})$ is Riemannian, then the orthogonal complement with respect to $\mathsf{h}$ is called the \emph{horizontal distribution} and denoted $\mathrm{Hor}$.
Typically, the horizontal distribution is \emph{not} integrable.
The Riemannian metric $\mathsf{h}$ is called \emph{descending} (with respect to the submersion $\pi$) if there exists a Riemannian metric $\mathsf{g}$ on $M$ such that for all $\z\in N$
\begin{equation}\label{eq:descending_metric}
	\mathsf{h}_{\z}(\vect u,\vect v) = \mathsf{g}_{\pi(\z)}(T_{\z}\pi\cdot \vect u,T_{\z}\pi\cdot \vect v), \quad \forall \vect u,\vect v\in \mathrm{Hor}_{\z}.
\end{equation}
The map $\pi$ between the Riemannian manifolds $(N,\mathsf{h})$ and $(M,\mathsf{g})$ is then called a \emph{Riemannian submersion}.
For details on Riemannian submersions, see~\cite[\S\!~3.5]{Pe2006}.

A vector field $Y\in\Xcal(N)$ is called \emph{horizontal} if $Y(\z)\in \mathrm{Hor}_{\z}$ for all $\z\in N$.
$Y$ is called \emph{descending} if there exists a vector field $X\in \Xcal(M)$ such that $\pi$ intertwines $X$ and $Y$, i.e., $T_{\z}\pi\cdot Y(\z) = X(\pi(\z))$.
The following result, schematically illustrated in \autoref{fig:riemann_midpoint}, is a generalized version of \autoref{pro:metric_equivariance}.

\begin{figure}
	\centering
	\def\svgwidth{0.45\textwidth}
	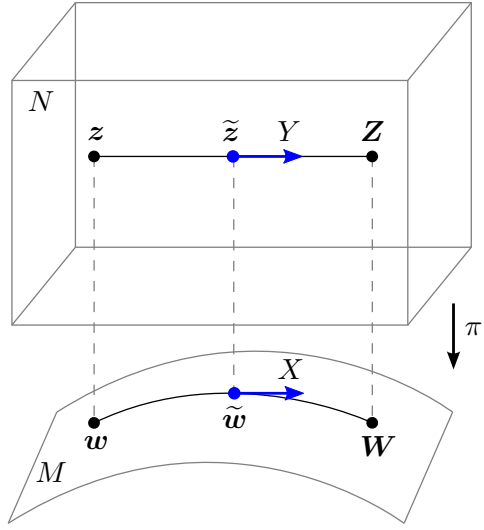
	\caption{
		An illustration of \autoref{pro:Riemannian_midpoint_descend}.
		If the vector fields $Y$ and $X$ are intertwined by the Riemannian submersion~$\pi$ and $Y$ is horizontal, then the Riemannian midpoint methods for~$X$ and~$Y$ are also intertwined by~$\pi$.
	}\label{fig:riemann_midpoint}
\end{figure}

\begin{theorem}\label{pro:Riemannian_midpoint_descend}
	Let $(M,\mathsf{g})$ and $(N,\mathsf{h})$ be Riemannian manifolds and $\pi\colon N \to M$ a Riemannian submersion.
	Let $Y\in\Xcal(N)$ be horizontal, let $X\in \Xcal(M)$, and assume that~$\pi$ intertwines $X$ and $Y$.
	Then $\pi$ intertwines the Riemannian midpoint method $\Phi_{\mathsf{g}}(hX)$ and the Riemannian midpoint method $\Phi_{\mathsf{h}}(\Delta t Y)$, i.e.,
	\begin{equation}\label{eq:midpoint_descend}
		\pi\big(\Phi_{\mathsf{h}}(\Delta tY)(\z)\big) = \Phi_{\mathsf{g}}(\Delta tX)(\pi(\z)).
	\end{equation}
\end{theorem}

\begin{proof}
	Let $\zminus$ and $\zplus$ fulfil~\eqref{eq:Riemannian_midpoint_method}.
	Let $\wminus = \pi(\zminus)$ and $\wplus = \pi(\zplus)$.
	We need to show that $\wplus = \Phi_{\mathsf{g}}(\Delta t X)(\wminus)$.
	The geodesic $\gamma_{\mathsf{h}}(t;\zminus,\zplus)$ is horizontal at $t=1/2$.
	It is therefore horizontal at all times~\cite[\S\!~3.5]{Pe2006}.
	Since horizontal geodesics on $N$ maps to geodesics on~$M$, we have that $\pi(\gamma_\mathsf{h}(t;\zminus,\zplus)) = \gamma_\mathsf{g}(t;\wminus,\wplus)$.
	By applying $T\pi$ to~\eqref{eq:Riemannian_midpoint_method} we obtain
	\begin{equation}
		\begin{split}
		T_{\gamma_{\mathsf{h}}(1/2,\zminus,\zplus)}\pi \cdot \frac{\ud}{\ud t}\Big|_{t=1/2} \gamma_{\mathsf{h}}(t;\zminus,\zplus) &=
			T_{\gamma_{\mathsf{h}}(1/2,\zminus,\zplus)}\pi \cdot \Delta t Y(\gamma_{\mathsf{h}}(1/2;\zminus,\zplus)) 
		\\ & \Rightarrow \\
		\frac{\ud}{\ud t}\Big|_{t=1/2} \pi\big(\gamma_{\mathsf{h}}(t;\zminus,\zplus)\big) &=
			\Delta t T_{\gamma_{\mathsf{h}}(1/2,\zminus,\zplus)}\pi \cdot Y(\gamma_{\mathsf{h}}(1/2;\zminus,\zplus))				
		\\ & \Rightarrow \\
		\frac{\ud}{\ud t}\Big|_{t=1/2} \pi\big(\gamma_{\mathsf{h}}(t;\zminus,\zplus)\big) &=
			\Delta t X\big(\pi(\gamma_{\mathsf{h}}(1/2;\zminus,\zplus))\big)				
		\\ & \Rightarrow \\
		\frac{\ud}{\ud t}\Big|_{t=1/2} \gamma_{\mathsf{g}}(t;\wminus,\wplus) &=
			\Delta t X\big(\gamma_{\mathsf{g}}(1/2;\wminus,\wplus)\big).
		\end{split}
	\end{equation}
	Thus, $\wplus$ fulfils the equation defining $\Phi_{\mathsf{g}}(\Delta t X)$, which proves the result.
\end{proof}

Next, assume $(N,\mathsf{h},\omega)$ is a Kähler manifold.
Recall the properties of a Kähler manifold: there is a map $J\colon TN\to TN$ called the \emph{complex structure} that fulfills
\begin{equation}\label{eq:Jprop}
	\begin{split}
	\mathsf{h}(\vect u,\vect v) &= \omega(\vect u,J\vect v) \\
	\omega(\vect u,\vect v) &= \mathsf{h}(J\vect u,\vect v) \\
	\mathsf{h}(J\vect u,J\vect v) &= \mathsf{h}(\vect u,\vect v) \\
	X_{F} &= J^{-1}\nabla F,\quad F\in C^{\infty}(N)
	\end{split}
\end{equation}
where $\nabla$ is the gradient with respect to $\mathsf{h}$.
We are interested in the case when $\pi$ is both a Riemannian submersion and a Poisson map.

\begin{lemma}\label{lem:horisontal_XH}
	Let $(N,\mathsf{h},\omega)$ be a Kähler manifold,
	let $(M,\mathsf{g},\{\cdot,\cdot\})$ be a Riemannian and Poisson manifold, and
	let $H\in C^\infty(M)$ be a Hamiltonian.
	Assume there is a Riemannian submersion $\pi\colon N\to M$ that is also a Poisson map.
	Then $X_{H\circ \pi}$ is horizontal if and only if $\nabla H$ (gradient of $H$ with respect to the metric on $M$) is tangent to the symplectic leaves of~$M$.
\end{lemma}

\begin{proof}
	By definition, the vector field $X_{H\circ\pi}$ is horizontal if and only if
	\begin{equation}\label{eq:XH_horisontal}
		\mathsf{h}(X_{H\circ\pi},\vect v) = 0, \quad \forall\, \vect v\in \mathrm{Vert}.
	\end{equation}
	By \eqref{eq:Jprop} we also have
	\begin{equation}\label{eq:g_omega_transforms}
		\begin{split}
		\mathsf{h}(X_{H\circ\pi},\vect v) &= 	\mathsf{h}(J^{-1}\nabla (H\circ\pi),\vect v) \\
		&= \mathsf{h}(\nabla (H\circ\pi),J\vect v)\\
		&= \omega(\vect v,\nabla (H\circ\pi))	.	
		\end{split}
	\end{equation}
	Combining \eqref{eq:XH_horisontal} and \eqref{eq:g_omega_transforms}, $X_{H\circ\pi}$ is horizontal if and only if
	\begin{equation}\label{eq:dH_omega_complement}
		\omega(\vect v,\nabla (H\circ\pi)) = 0, \quad \forall\, \vect v\in \mathrm{Vert}.
	\end{equation}
	Expressed in words, $X_{H\circ\pi}$ is horizontal if and only if $\nabla (H\circ\pi)$ belongs to the symplectic complement of $\mathrm{Vert}$, denoted $\mathrm{Vert}^{\bot_\omega}$.
	Let $\mathrm{Pre}$ denote the distribution on~$N$ defined by the preimage of the tangent spaces of the symplectic leaves of $M$.
	From \cite[Proposition~III.14.21]{LiMa1987} it follows that $\mathrm{Pre} = \mathrm{Vert}^{\bot_\omega}$.
	It remains to show that $\nabla H$ is tangent to the symplectic leaves if and only if $\nabla (H\circ\pi)\in \mathrm{Pre}$.
	Since the metric $\mathsf{h}$ is descending, the gradients on $M$ and $N$ are related by
	\begin{equation}\label{eq:pi_related_nabla}
		T\pi\circ \nabla (H\circ \pi) = \nabla H.
	\end{equation}
	By the definition of $\mathrm{Pre}$, this formula proves the result.
\end{proof}

By combining \autoref{pro:Riemannian_midpoint_descend} with \autoref{lem:horisontal_XH} we obtain the following result.

\begin{theorem}\label{pro:descending_midpoint_poisson}
	Let $M$, $N$, and $\pi$ be as in \autoref{lem:horisontal_XH}.
	Let $H\in C^{\infty}(M)$ fulfil the condition in \autoref{lem:horisontal_XH}, i.e., $\nabla H$ is tangent to the symplectic leaves.
	Let $\Phi_{\mathsf{g}}$ and $\Phi_{\mathsf{h}}$ denote the Riemannian midpoint methods on $M$ and $N$ respectively.
	Then:
	\begin{enumerate}
		\item If $\Phi_{\mathsf{h}}(X_{H\circ\pi})$ is symplectic, then $\Phi_{\mathsf{g}}(X_{H})$ is a Poisson map.
		\item If $\Phi_{\mathsf{h}}(X_{H\circ\pi})$ preserves the pre-image of the symplectic leaves, then $\Phi_{\mathsf{g}}(X_{H})$ preserves the symplectic leaves.
		\item If $G$ is a Lie group that acts on $M$ and $N$, and $\Phi_{\mathsf{h}}$ and $\pi$ are equivariant with respect to $G$, then $\Phi_{\mathsf{g}}$ is equivariant with respect to $G$.
	\end{enumerate}
\end{theorem}

\begin{proof}
	Let $\Phi_\mathsf{h}$ and $\Phi_{\mathsf{g}}$ be the Riemannian midpoint methods on $M$ and $N$ respectively.
	By \autoref{lem:horisontal_XH}, $X_{H\circ\pi}$ is horizontal, since $\nabla H$ is tangential to the symplectic leaves.
	The vector field $X_{H\circ\pi}$ is thus descending (it descends to $X_H$ since $\pi$ is a Poisson submersion) and horizontal. 
	By \autoref{pro:Riemannian_midpoint_descend}, $\pi$ then intertwines $\Phi_{\mathsf{g}}(X_{H})$ and $\Phi_{\mathsf{h}}(X_{H\circ\pi})$.
	
	Proof of (1):
	Since $\pi$ is a Poisson map and $\Phi_{\mathsf{h}}(X_{H\circ\pi})$ is a symplectic map, $\Phi_{\mathsf{g}}(X_{H})$ is a Poisson map.
	
	Proof of (2):
	Since $\Phi_{\mathsf{h}}(X_{H\circ\pi})$ preserves the integral submanifolds of the symplectic complement of $\mathrm{Vert}$, and since these submanifolds project to the symplectic leaves, it follows from the $\pi$ intertwining property that $\Phi_{\mathsf{g}}(X_{H})$ preserves the symplectic leaves.

	Proof of (3): Let $X\in \Xcal(M)$.
	Let $Y\in \Xcal(N)$ be horizontal and descending to $X$.
	Let $g\in G$
	Then $g\cdot X\circ \pi = g\cdot T\pi\circ X = T\pi g\cdot X$, since $\pi$ is equivariant.
	Thus, $g\cdot Y$ descends to $g\cdot X$.
	Next, using \autoref{pro:Riemannian_midpoint_descend}
	\begin{equation}
		\Phi_{\mathsf{g}}(g\cdot X)\circ\pi = \pi\circ \Phi_{\mathsf{h}}(g\cdot Y) = \pi \circ g^{-1}\cdot \Phi_{\mathsf{h}}(Y)\cdot g = g^{-1}\cdot \Phi_{\mathsf{g}}(X) \cdot g \circ \pi.
	\end{equation}
	This proves the results since $\pi$ is a submersion.
\end{proof}

\subsection{Riemannian structure of the spherical midpoint method} 
\label{sub:special_case}

Our objective is to show that the spherical midpoint method~\eqref{eq:area_midpoint_LP} on $\Hb*^n$, for Hamiltonian vector fields $X_H\in\Xcal(\Hb*^n)$ with $H$ of the form in \autoref{lem:horisontal_XH}, is a Riemannian midpoint method with respect to a non-Euclidean Riemannian metric, related to the classical midpoint method on $\HH*^n$ by a Riemannian submersion in the sense of \autoref{pro:Riemannian_midpoint_descend}.

Recall that the extended Hopf map~\eqref{eq:extended_Hopf_quaternions} is a submersion $\hopf\colon \HH*^n\to \Hb*^n$ that is a Poisson map with respect to the Kähler structure on $\HH*^n$ and the Poisson structure on $\Hb*^n$ (as described in~\autoref{sec:proof_2}).

\begin{lemma}\label{lem:Hopf_decending_metric}
	The Kähler metric on $\HH*^n$ is descending with respect to the extended Hopf map $\hopf\colon \HH*^n\to \Hb*^n$.
	The corresponding Riemannian metric on $\Hb*^n$ is
	\begin{equation}\label{eq:metric_on_N}
		\mathsf{g}_{\w}\big(\vect u,\vect v\big) \coloneqq \sum_{i=1}^{n} \frac{\vect u_i\cdot\vect v_i}{\abs{\w_i}}.
	\end{equation}
\end{lemma}

\begin{proof}
	Each fibre $\hopf^{-1}(\{ \w \} )\subset \HH*^n$ is the orbit of an action of the group $U(1)^n$ on $\HH*^n$.
	This action is isometric with respect to the Kähler metric.
	That is, if $\mathsf{h}$ denotes the Kähler metric and $L_{\theta}$ denotes the action map, then $L_{\theta}^{*}\mathsf{h} = \mathsf{h}$.
	It follows from \cite[Proposition~4.3]{Mo2015} that $\mathsf{h}$ is descending.
	Direct calculations, straightforward but lengthy, confirm that it descends to the metric~\eqref{eq:metric_on_N}.
\end{proof}

As a specialization of \autoref{pro:Riemannian_midpoint_descend} to the case $M=\Hb*^n$ and $N=\HH*^n$, we obtain a relation between the Riemannian midpoint method on $\Hb*^n$ and $\HH*^n$.
Notice that the Riemannian midpoint method on $\HH*^n$ is the classical midpoint method, since the metric of $\HH*^n$ is Euclidean.
We denote the classical midpoint method on $\HH*^n$ by $\Phi_{\textrm{cm}}$.

\begin{theorem}\label{pro:relation_midpoint_methods_up_down}
	Let $Y\in\Xcal(\HH*^n)$ be a horizontal vector field, let $X\in\Xcal(\Hb*^n)$, and assume that the extended Hopf map $\pi$ intertwines $X$ and $Y$.
	Further, let $\mathsf{g}$ denote the Riemannian metric~\eqref{eq:metric_on_N} on $\Hb*^n$.
	Then $\pi$ intertwines the Riemannian midpoint method $\Phi_{\mathsf{g}}(hX)$ and the classical midpoint method $\Phi_{\textrm{cm}}(hY)$.
\end{theorem}

As a specialization of \autoref{lem:horisontal_XH} to the case $M=\Hb*^n$ and $N=\HH*^n$, we obtain a geometric formulation of \autoref{lma:constantorthogonal}.

\begin{lemma}\label{lem:horisontal_XH_R4n}
	Let $H\in C^\infty(\Hb*^n)$.
	Then $X_{H\circ\hopf}$ is horizontal if and only if $H$ is constant on the rays.
	In particular, $X_{H \circ \projS \circ \hopf}$ is horizontal for any $H \in C^\infty((S^2)^n)$.
\end{lemma}

\begin{proof}
	The symplectic leaves of $\Hb*^n$ are the coadjoint orbits of $(\so(3)^*)^{n}$.
	These consists of
	\begin{equation}\label{eq:symplectic_leaves_R3n}
		S^2_{r_1}\times\cdots\times S^2_{r_n} = \{ (\w_{1},\ldots,\w_{n})\in \Hb*^n; \abs{\w_k} = r_k \},
	\end{equation}
	for arbitrary $r_k \in \RR^{+}$.
	Let $\nabla$ denote the gradient on $\Hb*^n$ with respect to $\mathsf{g}$.
	It follows from \autoref{lem:horisontal_XH} that $X_{H\circ\pi}$ is horizontal if and only if $\nabla H$ is tangent to \eqref{eq:symplectic_leaves_R3n}.
	From~\eqref{eq:metric_on_N} and the direct product structure of $\Hb*^n$, we see that $\mathsf{g}(\vect u,\vect v) = 0$ for any $\vect u$ tangent to \eqref{eq:symplectic_leaves_R3n} if and only if $\vect v$ is tangent to the rays (\autoref{def:rays}).
	Let $R_{\projS(\w)}\subset \Hb*^n$ denote the ray through $\w$.
	The condition for $X_{H\circ\hopf}$ to be horizontal is therefore
	\begin{equation}\label{eq:cond_XHopi_R4n_horiz}
		\begin{split}
			\mathsf{g}_{\w}(\nabla H(\w),\vect v) &= 0,\quad
			\forall \vect v \in T_{\w}R_{\projS(\w)}
		\\ &\iff 
		\\
			\pair{\ud H(\w)}{\vect v} &= 0,\quad
			\forall \vect v \in T_{\w}R_{\projS(\w)},
		\end{split}
	\end{equation}
	which implies that $H$ must be constant on the rays.
\end{proof}

As a specialization of \autoref{pro:descending_midpoint_poisson} to the case $M=\Hb*^n$ and $N=\HH*^n$, we recover again that the Riemannian midpoint method $\Phi_{\mathsf{g}}$ on $\Hb*^n$ is a Poisson integrator.

\begin{proposition}\label{pro:poisson_strange_midpoint_method}
	The Riemannian midpoint method $\Phi_{\mathsf{g}}$ on $\Hb*^n$, applied to Hamiltonian vector fields, is a Poisson integrator.
\end{proposition}

The final result in this section connects the extended spherical midpoint method~\eqref{eq:area_midpoint_LP} and the Riemannian midpoint method $\Phi_{\mathsf{g}}$.
The two methods are different, but they coincide for ray-constant Hamiltonian vector fields, as a consequence of \autoref{thm:midpoint_relatedness} and \autoref{pro:relation_midpoint_methods_up_down}.

\begin{proposition}\label{pro:relation_strange_spherical}
	Let $\Psi$ denote the extended spherical midpoint method~\eqref{eq:area_midpoint_LP} on $\Hb*^n$ and $\Phi_{\mathsf{g}}$ the Riemannian midpoint method with respect to the metric $\mathsf{g}$ in \eqref{eq:metric_on_N}.
	Let $H\in C^{\infty}(\Hb*^n)$ be constant on rays.
	Then $\Psi(\Delta tX_{H}) = \Phi_{\mathsf{g}}(\Delta t X_{H})$.
\end{proposition}

\autoref{pro:poisson_strange_midpoint_method} and \autoref{pro:relation_strange_spherical} provide another independent proof that the spherical midpoint method is a symplectic discretization, this time based on Kähler geometry and \autoref{lem:horisontal_XH}.
This proof is interesting because it allows for generalizations to other Kähler manifolds.

It is a remarkable consequence of \autoref{pro:relation_strange_spherical} that the non-Euclidean induced metric $\mathsf{g}$, given by~\eqref{eq:metric_on_N}, has become redundant in the case when the Hamiltonian $H$ is constant on rays; 
from~\eqref{eq:classical_spherical_relation} we see that the Riemannian midpoint method $\Phi_{\mathsf{g}}(\Delta tX_{H})$ can be expressed solely in terms of the classical midpoint method on $\Hb*^n$ for such Hamiltonians.

\section*{Acknowledgments}

The research was supported by the J~C~Kempe Memorial Fund, the Swedish Foundation for Strategic Research (ICA12-0052), EU Horizon 2020 Marie Sklodowska-Curie Individual Fellowship (661482), and the Marsden Fund of the Royal Society of New Zealand.


\IfFileExists{./amsplainnat.bst}{
\bibliographystyle{amsplainnat}
}{
\bibliographystyle{../amsplainnat}
}
\IfFileExists{./collective.bib}{
\bibliography{collective}
}{
\bibliography{../collective}
}

\end{document}